\documentclass[%
 reprint,
 onecolumn,
 amsmath,amssymb,
 aps,
]{revtex4-1}

\usepackage{graphicx}% Include figure files
\usepackage{dcolumn}% Align table columns on decimal point
\usepackage{bm}% bold math
\usepackage{amsmath,amsfonts,amssymb,dsfont,hyperref,color,epsfig,graphics,graphicx,latexsym,mathrsfs,revsymb,theorem,url,verbatim,epstopdf,cleveref,dashrule,cases,booktabs,diagbox,threeparttable,setspace,natbib,mathdots}

\hypersetup{colorlinks,linkcolor={blue},citecolor={blue},urlcolor={red}}

\newtheorem{definition}{Definition}

\newtheorem{theorem}[definition]{Theorem}

\def\squareforqed{$\square$}
\def\qed{\ifmmode\squareforqed\else{\unskip\nobreak\hfil
\penalty50\hskip1em\null\nobreak\hfil\squareforqed
\parfillskip=0pt\finalhyphendemerits=0\endgraf}\fi}

\def\endenv{\ifmmode\;\else{\unskip\nobreak\hfil
\penalty50\hskip1em\null\nobreak\hfil\;
\parfillskip=0pt\finalhyphendemerits=0\endgraf}\fi}

\newenvironment{proof}{\noindent \textbf{{Proof.~} }}{\qed}

\def\bpf{\begin{proof}}
\def\epf{\end{proof}}
\def\bea{\begin{eqnarray}}
\def\eea{\end{eqnarray}}
\def\beq{\begin{equation}}
\def\eeq{\end{equation}}
\def\bal{\begin{aligned}}
\def\eal{\end{aligned}}
\def\bma{\begin{bmatrix}}
\def\ema{\end{bmatrix}}

\def\min{\mathop{\rm min}}

\def\ra{\rightarrow}

\def\ox{\otimes}

\def\a{\alpha}

\def\d{\delta}

\def\m{\mu}
\def\n{\nu}

\def\p{\pi}
\def\r{\rho}
\def\s{\sigma}

\def\G{\Gamma}

\newcommand{\bra}[1]{\langle{#1}|}
\newcommand{\ket}[1]{|{#1}\rangle}
\newcommand{\proj}[1]{|{#1}\rangle \langle {#1}|}
\newcommand{\ketbra}[2]{|{#1}\rangle \! \langle{#2}|}

\newcommand{\nc}{\newcommand}

\nc{\bbA}{\mathbb{A}} \nc{\bbB}{\mathbb{B}} \nc{\bbC}{\mathbb{C}}
\nc{\bbD}{\mathbb{D}} \nc{\bbE}{\mathbb{E}} \nc{\bbF}{\mathbb{F}}
\nc{\bbG}{\mathbb{G}} \nc{\bbH}{\mathbb{H}} \nc{\bbI}{\mathbb{I}}
\nc{\bbJ}{\mathbb{J}} \nc{\bbK}{\mathbb{K}} \nc{\bbL}{\mathbb{L}}
\nc{\bbM}{\mathbb{M}} \nc{\bbN}{\mathbb{N}} \nc{\bbO}{\mathbb{O}}
\nc{\bbP}{\mathbb{P}} \nc{\bbQ}{\mathbb{Q}} \nc{\bbR}{\mathbb{R}}
\nc{\bbS}{\mathbb{S}} \nc{\bbT}{\mathbb{T}} \nc{\bbU}{\mathbb{U}}
\nc{\bbV}{\mathbb{V}} \nc{\bbW}{\mathbb{W}} \nc{\bbX}{\mathbb{X}}
\nc{\bbY}{\mathbb{Y}} \nc{\bbZ}{\mathbb{Z}}

\nc{\bA}{{\bf A}} \nc{\bB}{{\bf B}} \nc{\bC}{{\bf C}}
\nc{\bD}{{\bf D}} \nc{\bE}{{\bf E}} \nc{\bF}{{\bf F}}
\nc{\bG}{{\bf G}} \nc{\bH}{{\bf H}} \nc{\bI}{{\bf I}}
\nc{\bJ}{{\bf J}} \nc{\bK}{{\bf K}} \nc{\bL}{{\bf L}}
\nc{\bM}{{\bf M}} \nc{\bN}{{\bf N}} \nc{\bO}{{\bf O}}
\nc{\bP}{{\bf P}} \nc{\bQ}{{\bf Q}} \nc{\bR}{{\bf R}}
\nc{\bS}{{\bf S}} \nc{\bT}{{\bf T}} \nc{\bU}{{\bf U}}
\nc{\bV}{{\bf V}} \nc{\bW}{{\bf W}} \nc{\bX}{{\bf X}}
\nc{\bY}{{\bf Y}} \nc{\bZ}{{\bf Z}}

\nc{\bmA}{{\bm A}} \nc{\bmB}{{\bm B}} \nc{\bmC}{{\bm C}}
\nc{\bmD}{{\bm D}} \nc{\bmE}{{\bm E}} \nc{\bmF}{{\bm F}}
\nc{\bmG}{{\bm G}} \nc{\bmH}{{\bm H}} \nc{\bmI}{{\bm I}}
\nc{\bmJ}{{\bm J}} \nc{\bmK}{{\bm K}} \nc{\bmL}{{\bm L}}
\nc{\bmM}{{\bm M}} \nc{\bmN}{{\bm N}} \nc{\bmO}{{\bm O}}
\nc{\bmP}{{\bm P}} \nc{\bmQ}{{\bm Q}} \nc{\bmR}{{\bm R}}
\nc{\bmS}{{\bm S}} \nc{\bmT}{{\bm T}} \nc{\bmU}{{\bm U}}
\nc{\bmV}{{\bm V}} \nc{\bmW}{{\bm W}} \nc{\bmX}{{\bm X}}
\nc{\bmY}{{\bm Y}} \nc{\bmZ}{{\bm Z}}

\nc{\cA}{{\cal A}} \nc{\cB}{{\cal B}} \nc{\cC}{{\cal C}}
\nc{\cD}{{\cal D}} \nc{\cE}{{\cal E}} \nc{\cF}{{\cal F}}
\nc{\cG}{{\cal G}} \nc{\cH}{{\cal H}} \nc{\cI}{{\cal I}}
\nc{\cJ}{{\cal J}} \nc{\cK}{{\cal K}} \nc{\cL}{{\cal L}}
\nc{\cM}{{\cal M}} \nc{\cN}{{\cal N}} \nc{\cO}{{\cal O}}
\nc{\cP}{{\cal P}} \nc{\cQ}{{\cal Q}} \nc{\cR}{{\cal R}}
\nc{\cS}{{\cal S}} \nc{\cT}{{\cal T}} \nc{\cU}{{\cal U}}
\nc{\cV}{{\cal V}} \nc{\cW}{{\cal W}} \nc{\cX}{{\cal X}}
\nc{\cY}{{\cal Y}} \nc{\cZ}{{\cal Z}}

\nc{\hA}{{\hat{A}}} \nc{\hB}{{\hat{B}}} \nc{\hC}{{\hat{C}}}
\nc{\hD}{{\hat{D}}} \nc{\hE}{{\hat{E}}} \nc{\hF}{{\hat{F}}}
\nc{\hG}{{\hat{G}}} \nc{\hH}{{\hat{H}}} \nc{\hI}{{\hat{I}}}
\nc{\hJ}{{\hat{J}}} \nc{\hK}{{\hat{K}}} \nc{\hL}{{\hat{L}}}
\nc{\hM}{{\hat{M}}} \nc{\hN}{{\hat{N}}} \nc{\hO}{{\hat{O}}}
\nc{\hP}{{\hat{P}}} \nc{\hR}{{\hat{R}}} \nc{\hS}{{\hat{S}}}
\nc{\hT}{{\hat{T}}} \nc{\hU}{{\hat{U}}} \nc{\hV}{{\hat{V}}}
\nc{\hW}{{\hat{W}}} \nc{\hX}{{\hat{X}}} \nc{\hZ}{{\hat{Z}}}

\nc{\hn}{{\hat{n}}}

\def\max{\mathop{\rm max}}
\def\min{\mathop{\rm min}}

\def\ox{\otimes}

\def\ra{\rightarrow}

\begin{document}

\preprint{APS/123-QED}

\title{Constructing unextendible product bases from multiqubit ones}

% Force line breaks with \\

%\affiliation{%

\author{Taiyu Zhang} 
\email[]{2241717357@qq.com}
\affiliation{159th Middle School of Beijing, Xuanwu Men West Street, Xicheng District, Beijing, 100031, China}

\author{Lin Chen}
\email[]{linchen@buaa.edu.cn(corresponding author)}
\affiliation{LMIB and School of Mathematical Sciences, Beihang University, Beijing 100191, China}
\affiliation{International Research Institute for Multidisciplinary Science, Beihang University, Beijing 100191, China}

\date{\today}% It is always \today, today,

\Large

\begin{abstract}
The construction of multipartite unextendible product bases (UPBs) is a basic problem in quantum information. We respectively construct two families of $2\times2\times4$ and $2\times2\times2\times4$ UPBs of size eight by using the existing four-qubit and five-qubit UPBs. As an application, we construct novel families of multipartite positive-partial-transpose entangled states, as well as their entanglement properties in terms of the geometric measure of entanglement. 
\end{abstract}

                             
\maketitle

%\tableofcontents

\section{Introduction}
\label{sec:int}

The unextendible product bases (UPBs) have been applied to various quantum-information processing in the past decades. First, UPBs help construct positive-partial-transpose (PPT) entangled states
\cite{BDF+99,dms03,Pittenger2003Unextendible,Terhal2001A,SLJM10}, the nonlocality without entanglement and Bell inequalities  \cite{ ASH+11,AL01,Fen06,Chen2013The, Tura2012Four,Szanto2016Complementary,Chen2014Unextendible,Hou2014A,2021Strongly,Shi2021}, genuine entanglement \cite{PhysRevA.98.012313}, and Schmidt number \cite{PhysRevA.90.054303}. 
%Scope of cardinalities of UPBs \cite{AL01,Fen06,Chen2013The,}
 %no relation \cite{Tura2012Four,Szanto2016Complementary,Chen2014Unextendible,Hou2014A}. 
Multipartite UPBs have been related to the tile structures and local entanglement-assisted distinguishability
\cite{2020Unextendible}. Further, multiqubit UPBs have received extensive attention \cite{Reed2012Realization,Joh13} due to the role of multiqubit system in quantum computing and  experiments \cite{Dicarlo2010Preparation}. For example, the three and four-qubit UPBs have been fully studied in terms of programs 
 \cite{Bravyi2004Unextendible,Johnston2014The}. Some seven-qubit UPBs of size ten have been
 constructed \cite{2020The}. Further the exclusion of $n$-qubit UPBs of size $2^n-5$  \cite{Chen2018Nonexistence} has solved an open problem in \cite{Johnston2014The}. Nevertheless, the understanding of multipartite UPBs is still far from a complete picture. In particular, 
 It is an intriguing problem to find the relation between multiqubit UPBs and multipartite UPBs in high dimensions, because the latter are usually harder to construct. This is the first motivation of this paper.
 
On the other hand, it is known that PPT entangled states are not distillable. That is, they cannot be asymptotically converted into Bell states under local operations and classical communications (LOCC), while Bell states are necessary for most quantum-information tasks such as teleportation and computing. In contrast, some PPT entangled states can be used to create distillable key \cite{hhh05}. the PPT entangled states are related to the entanglement distillation problem \cite{2021A,Sun2020}, the detection of entanglement \cite{Shen2020,2021Detection}, as well as multipartite genuinely entangled states and entanglement-breaking subspaces \cite{Sun2020-2,Shen2020-2}. 
The study of multipartite PPT states has also been devoted to the separability of completely symmetric states \cite{2019Separability}, the separability of symmetric states and vandermonde decomposition \cite{Qian2020}. Hence, it is an important problem to construct novel PPT entangled states from UPBs, which may evidently motivated the study of fore-mentioned topics. This is the second motivation of this paper.  

In this paper, we construct a family of $2\times2\times4$ UPBs of size eight by using the merge of some systems of existing four-qubit UPBs of size eight in Theorem \ref{thm:2x2x2x2->2x2x4}, as well as a family of $2\times2\times2\times4$ UPBs of size eight by using the existing five-qubit UPBs of size eight in Theorem \ref{thm:2x2x2x2x2->2x2x2x4}. In particular, we shall show that the resulting set by any one of the four merge AD, BC, BD, and CD in Eq. \eqref{eq:01} is not a $2\times2\times4$ UPB, and the resulting set by any one of the two merge AB and AC in \eqref{eq:01} is a $2\times2\times4$ UPB of size eight. Further, the resulting set by any one of the four merge AB, CD, CE and DE in \eqref{eq:04} is not a $2\times2\times2\times4$ UPB, and the resulting set by any one of the six merge AC, AD, AE, BC, BD and BE in \eqref{eq:04} is a $2\times2\times2\times4$ UPB of size eight. As an application, we construct two families of multipartite PPT entangled states, and evaluate their entanglement in terms of geometric measurement of entanglement. In particular, we construct the upper bound of one family of tripartite PPT entangled states in terms of its parameters. Our work shows the latest progress on the construction of multipartite UPBs and PPT entangled states by means of UPBs.

%\tbc \red{(Taiyu may do it if we still have time later.) We work on some more cases such as five-qubit UPB of size ten to explore whether the number of existing constructed UPB grows as the number of qubit grows.}

The rest of this paper is organized as follows. In Sec. \ref{sec:pre} we introduce the preliminary knowledge used in this paper. In Sec. \ref{sec:res224} and \ref{sec:res2224}, we respectively construct a family of $2\times2\times4$ and $2\times2\times2\times4$ UPBs. Using them, we establish two families of PPT entangled states, and explore the properties of their entanglement in Sec. \ref{sec:app}. Finally we conclude in Sec. \ref{sec:con}.

\section{Preliminaries}
\label{sec:pre}

In this section we introduce the preliminary knowledge used in this paper. We work with an $n$-partite quantum system $A_1,A_2,\cdots,A_n$ in the Hilbert space $\cH=\cH_1\ox\cdots\ox\cH_n=\bbC^{d_1}\otimes...\otimes\bbC^{d_n}$.
An $n$-partite product vector in 
$\cH$ is denoted as $\ket{\psi}=\ket{a_1}\ox\cdots\ox\ket{a_n}:=\ket{a_1,\ldots,a_n}\quad \text{with}\quad \ket{a_i}\in\cH_i$. Although the normalization factor is necessary for the interpretation of quantum states, we do not always normalize product vectors for the sake of mathematical convenience. Next, a set of $n$-partite orthonormal product vectors $\{\ket{a_{i,1}},...,\ket{a_{i,n}}\}$ is an {\em unextendible product basis} (UPB) in $\cH$, if there is no $n$-partite product vector orthogonal to all product vectors in the set at the same time. In particular, the $n$-qubit UPB exists when $d_i=2$ for any $i$. For simplicity when $n=5$, we refer to $A_1,A_2,A_3,A_4$ and $A_5$ as $A,B,C,D$ and $E$, respectively. We have
 $\bbC^2\ox\bbC^2\ox\bbC^2\ox\bbC^2\ox\bbC^2:= \cH_A\ox\cH_B\ox\cH_C\ox\cH_D\ox\cH_E$, and $\cH_{AB}:=\cH_A\ox\cH_B$, and so on. We take the vectors $\ket{0}:=\bma1\\0\ema$ and $\ket{1}:=\bma0\\1\ema$, so that the set $\{\ket{0}, \ket{1}\}$ is a qubit orthonormal basis in $\bbC^2$. We shall denote a general qubit orthonormal basis by $\{\ket{x},\ket{x'}\}$, with $x=a,b,c$ and so on. It is straightforward to verify that, a UPB remains a UPB if we switch the systems or perform any local unitary transformation \cite{2020The}. We say that two UPBs are equivalent if one UPB can be obtained from the other by the switch or local unitary transformation.

\section{Tripartite UPBs from four-qubit UPBs}
\label{sec:res224}

In this section, we construct a tripartite UPB of size eight, by using an existing four-qubit UPB $\cS$ of size eight in the space $\cH_A \otimes \cH_B \otimes \cH_C \otimes \cH_D$. The UPB $\cS$ was constructed in \cite{Johnston2014The}. For convenience, we describe the matrix of $\cS$ in \eqref{eq:00}. Note that the first row of \eqref{eq:00} means the product state $\ket{0,0,0,0}$ in $\cS$, and one can similarly figure out all elements in $\cS$. The matrix representation of a UPB has been used to characterize the four-qubit orthogonal product bases \cite{1751-8121-50-39-395301,Chen2018Multiqubit}, and four-qubit UPBs \cite{Chen2018The}. 
\begin{eqnarray}
\label{eq:00}
\bma
0 & 0 & 0 & 0\\
1 & a & a' & a \\
a & a' & 1 & a' \\
a' & 1 & a & b \\
0 & a & a' & 1 \\
1 & a & a' & a' \\
a & 1 & a & a \\
a' & a' & 1 & b' \\
\ema,
\\
\label{eq:01}
\bma
0 & 0 & 0 & 0\\
1 & a & a' & a\\
0 & a & a' & 1\\
1 & a & a' & a'\\
a & a' & 1 & a'\\
a' & 1 & a & b\\
a & 1 & a & a\\
a' & a' & 1 & b'\\
\ema.
\end{eqnarray}
By switching rows 3 and 5, and rows 4 and 6 in \eqref{eq:00}, we obtain \eqref{eq:01}.
We merge two of the systems A, B, C and D, so that \eqref{eq:01} corresponds to a set of $2\times2\times4$ orthonormal product vectors. For example, the merge of AB in \eqref{eq:01} implies the set $\{\ket{0,0}\ket{0}\ket{0},\ket{1,a}\ket{a'}\ket{a},...,\ket{a',a'}\ket{1}\ket{b}\}$. One can verify that there are six ways for the merge, namely AB, AC, AD, BC, BD, and CD. We present the following observation. 

\begin{theorem}
\label{thm:2x2x2x2->2x2x4}
The resulting set by any one of the four merge AD, BC, BD, and CD in \eqref{eq:01} is not a $2\times2\times4$ UPB. The resulting set by any one of the two merge AB and AC in \eqref{eq:01} is a $2\times2\times4$ UPB of size eight.
\end{theorem}
\begin{proof}
We prove the claim by six cases \textbf{(I)-(VI)}. They respectively work for the system merge CD, BD, BC, AD, AC and AB in \eqref{eq:01}. Using the merge, we shall refer to the set of orthonormal product vectors corresponding to \eqref{eq:01} as $\cS_1,...,\cS_6$, respectively, in the six cases \textbf{(I)-(VI)}. Further, if any two systems of $A,B,C,D$ merge then evidently they are in $\bbC^4$. For example, the merge of CD makes the set of four-qubit orthonormal vectors corresponding to \eqref{eq:01} become $2\times2\times4$ orthonormal vectors in $\bbC^2\otimes\bbC^2\otimes\bbC^4=\cH_A\otimes\cH_B\otimes\cH_{CD}$.

(I)	When we merge the system CD, one can verify that the set $\cS_1$ in \eqref{eq:01} is orthogonal to $\ket{a, a', x_1}_{A: B: CD} \in \cH_A \otimes \cH_B \otimes \cH_{CD}$, where $\ket{x_1}$ is orthogonal to $\ket{0,0}$, $\ket{1,a'}$ and $\ket{a,a}$. So $\cS_1$ is not a UPB in $\cH_A \otimes \cH_B \otimes \cH_{CD}$.

(II) When we merge the system BD, one can verify that the set $\cS_2$ in \eqref{eq:01} is orthogonal to $\ket{a', a, x_2}_{A: C: BD} \in \cH_A \otimes \cH_C \otimes \cH_{BD}$, where $\ket{x_2}$ is orthogonal to $\ket{0,0}$, $\ket{1,b}$ and $\ket{a',b'}$. So $\cS_2$ is not a UPB in $\cH_A \otimes \cH_C \otimes \cH_{BD}$.

(III) When we merge the system BC, one can verify that the set $\cS_3$ in \eqref{eq:01} is orthogonal to $\ket{1, a, x_3}_{A: D: BC} \in \cH_A \otimes \cH_D \otimes \cH_{BC}$, where $\ket{x_3}$ is orthogonal to $\ket{a,a'}$, $\ket{1,a}$ and $\ket{a',1}$. So $\cS_3$ is not a UPB in $\cH_A \otimes \cH_D \otimes \cH_{BC}$.

(IV) When we merge the system AD, suppose the set $\cS_4$ in \eqref{eq:01} is orthogonal to $\ket{0, a, x_4}_{B: C: AD} \in \cH_B \otimes \cH_C \otimes \cH_{AD}$, one can verify that $\ket{x_4}$ is orthogonal to $\ket{0,0}$, $\ket{a,a'}$ and $\ket{a',b'}$. So $\cS_4$ is not a UPB in $\cH_B \otimes \cH_C \otimes \cH_{AD}$.

(V)	Before carrying out the proof, we label $a,a'$ in column three and four of \eqref{eq:01} as $a_3,a_3'$ and $a_4,a_4'$, respectively. To prove the assertion, it suffices to find some $2\times2\times4$ UPBs in \eqref{eq:01} by merging system AC. So 
we construct the following expressions related to \eqref{eq:01}.
\begin{eqnarray}
\label{eq:ket0}
&&
\ket{0} = \bma 1 \\ 0 \ema, 
\quad
\ket{1} = \bma 0 \\ 1 \ema,
\\&&
\ket{a_3} = \bma \cos x_3 \\ \sin x_3 \ema,
\quad
\ket{a_3'} = \bma \sin x_3 \\ -\cos x_3 \ema,
\\&&
\ket{a_4} = \bma \cos x_4 \\ \sin x_4 \ema,
\quad
\ket{a_4'} = \bma \sin x_4 \\ -\cos x_4 \ema, 
\\&&
\ket{c_1} = \ket{c_2} =\ket{1, a_4'} = \bma 0 \\ 0 \\ \sin x_4 \\ -\cos x_4 \ema,
\\&&
\ket{c_3} = \ket{a_3, 1} = \bma 0 \\ \cos x_3 \\ 0 \\ \sin x_3 \ema,
\\&&
\ket{c_4} = \ket{a_3', a_4} = \bma \sin x_3 \cos x_4  \\ \sin x_3 \sin x_4 \\ -\cos x_3 \cos x_4 \\ -\cos x_3  \sin x_4 \ema,
\\&&
\ket{c_5} = \ket{0, a_4'} = \bma \sin x_4  \\ -\cos x_4 \\ 0 \\ 0 \ema,
\\&&
\ket{c_6} =\ket{0, 0} = \bma 1 \\ 0 \\ 0 \\ 0 \ema,
\\&&
\ket{c_7} = \ket{a_3, a_4} = \bma \cos x_3 \cos x_4 \\ \cos x_3 \sin x_4  \\ \sin x_3 \cos x_4  \\ \sin x_3 \sin x_4  \ema,
\\&&
\ket{c_8} = \ket{a_3', 1} = \bma 0 \\ \sin x_3  \\ 0 \\ -\cos x_3  \ema.
\label{eq:c8}
\end{eqnarray}
Using these equations, we construct the $4\times7$ matrix $(\ket{c_2},\ket{c_3}, ......,\ket{c_8})$ as follows.
\begin{eqnarray}
\bma
0 & 0 & \sin x_3 \cos x_4 & \sin x_4 & 1 & \cos x_3 \cos x_4 & 0\\
0 & \cos x_3 & \sin x_3 \sin x_4 & -\cos x_4 & 0 & \cos x_3 \sin x_4 & \sin x_3\\
\sin x_4 & 0 & -\cos x_3 \cos x_4 & 0 & 0 & \sin x_3 \cos x_4 & 0\\
-\cos x_4 & \sin x_3 & -\cos x_3 \sin x_4 & 0 & 0 & \sin x_3 \sin x_4 & -\cos x_3\\
\ema.
\end{eqnarray}
We merge the system A,C, and suppose that the set $\cS_5=\{\ket{a_j,b_j,c_j},j=1,...,8\}$ in \eqref{eq:01} is orthogonal to 
\begin{eqnarray}
\label{eq:uvx5}
\ket{u, v, x_5}_{B: D: AC} \in \cH_B \otimes \cH_D \otimes \cH_{AC}. \end{eqnarray}

Hence $\ket{u,v}$ is orthogonal to some of the eight two-qubit product vectors $\ket{a_1,b_1}$,…,$\ket{a_8,b_8}$. By checking the expression of \eqref{eq:01}, one can show that each column of \eqref{eq:01} has at most three identical elements. In particular, the set $\{\ket{a_j}\}$ has at most three identical elements, and the set $\{\ket{b_j}\}$ has at most two identical elements. As a result, $\ket{u, v}$ is orthogonal to at most five of the eight product vectors $\ket{a_1,b_1}$,…,$\ket{a_8,b_8}$. We discuss three cases (V.a), (V.b) and (V.c). 

(V.a) Up to the permutation of subscripts in $\cS_5$, we may assume that $\ket{u, v}$ is orthogonal to exactly five product vectors $\ket{a_1, b_1}$, $\ket{a_2, b_2}$,..., $\ket{a_5, b_5}$. That is, $\ket{u}$ is orthogonal to $m$ of the five vectors $\ket{a_1}$, $\ket{a_2}$, ..., $\ket{a_5}$, e.g. $\ket{a_1}$, ..., $\ket{a_m}$ and $\ket{v}$ is orthogonal to $\ket{b_{m+1}}$, ..., $\ket{b_5}$. Using the property of set $\cS_5$ in \eqref{eq:01}, one can show that $m=3$, that is, $\ket{u}$ is orthogonal to three identical vectors $\ket{a_1}, \ket{a_2}$ and $\ket{a_3} \in \cH_B$, $\ket{v}$ is orthogonal to two identical vectors $\ket{b_3}$ and $\ket{b_4} \in \cH_D$. However the product vectors $\ket{a_1, b_1}$, ..., $\ket{a_4, b_4}$ do not occupy five rows of the matrix of $\cS$. So $\ket{u, v}$ is not orthogonal to five product vectors in $\ket{a_j, b_j}$, which means this case is not possible.

(V.b) Up to the permutation of subscripts in $\cS_5$, we may assume that $\ket{u, v}$ is orthogonal to exactly four product vectors $\ket{a_1, b_1}$, $\ket{a_2, b_2}$, $\ket{a_3, b_3}$, and $\ket{a_4, b_4}$. So $\ket{x_5}$ is orthogonal to four of the eight product vectors $\ket{c_1}$, ... , $\ket{c_8}$. By observing them, one can assume that $\ket{c_1}$=$\ket{c_2}$. Let the four product vectors be $\ket{c_{j_1} }$, $\ket{c_{j_2} }$, $\ket{c_{j_3} }$, $\ket{c_{j_4} }$, respectively, where $j_1, j_2, j_3, j_4$ are distinct integers in [2,8]. So there are 35 distinct matrices, which are formed by choosing the arrays $(j_1, j_2, j_3, j_4)$ as $(2,3,4,5), (2,3,4,6)$,...,$(5,6,7,8)$, respectively. Because $\ket{x_5}$ is orthogonal to $\ket{c_{j_1} }$, $\ket{c_{j_2} }$, $\ket{c_{j_3} }$, $\ket{c_{j_4} }$, we obtain that the matrices have determinant zero. According to our calculation, there are exactly three matrices of determinant zero, whose arrays are (2,3,5,7), (2,4,5,8), (3,5,6,8), respectively. However by checking $\cS_5$ in \eqref{eq:01}, one can see that for each array, there is no $\ket{u,v,x_5}$ orthogonal to $\cS_5$. It is a contradiction with \eqref{eq:uvx5}.

(V.c) Up to the permutation of subscripts in $\cS_5$, we may assume that $\ket{u, v}$ is orthogonal to at most three product vectors $\ket{a_1, b_1}$, $\ket{a_2, b_2}$, $\ket{a_3, b_3}$. So $\ket{x_5}$ is orthogonal to five of the eight product vectors $\ket{c_1}$, ... , $\ket{c_8}$. By observing them, one can assume that $\ket{c_1}=\ket{c_2}$. Let the five product vectors be $\ket{c_{j_1}}$, $\ket{c_{j_2}}$, $\ket{c_{j_3}}$, $\ket{c_{j_4}}$, $\ket{c_{j_5}}$, respectively, where $j_1, j_2, j_3, j_4, j_5$ are distinct integers in [2,8]. Because they are orthogonal to $\ket{x_5}$, they are linearly dependent. Hence, any four of $\ket{c_{j_1}}$, $\ket{c_{j_2}}$, $\ket{c_{j_3}}$, $\ket{c_{j_4}}$, $\ket{c_{j_5}}$ form a 4x4 matrix of determinant zero. So at least five of the 35 matrices in (V.b) have determinant zero. It is a contradiction with the calculation in (V.b), namely there are only three matrices of determinant zero. Hence this case is not possible. 

To conclude, we have excluded the three cases (V.a), (V.b), and (V.c). So the set $\cS_5=\{\ket{a_j,b_j,c_j},j=1, ..., 8\}$ in \eqref{eq:01} is a $ 2 \times 2 \times 4$ UPB in $\cH_B \otimes \cH_D \otimes \cH_{AC}$.

(VI) We convert the set $\cS_6$ in \eqref{eq:01} in the space $\cH_C \otimes \cH_D \otimes \cH_{AB}$ into the UPB $\cS_5$ in case (V), by switching column 2 and 3 and a series of row permutations on the matrix \eqref{eq:01}. It implies that the matrix $\cS$ is also a UPB $\in \cH_C \otimes \cH_D \otimes \cH_{AB}$. We list the switch and permutations on \eqref{eq:01} as follows. 
\begin{eqnarray}
\bma
0 & 0 & 0 & 0\\
1 & a_2 & a_3' & a_4\\
0 & a_2 & a_3' & 1\\
1 & a_2 & a_3' & a_4'\\
a_1 & a_2' & 1 & a_4'\\
a_1' & 1 & a_3 & b_4\\
a_1 & 1 & a_3 & a_4\\
a_1' & a_3' & 1 & b_4'\\
\ema,
\end{eqnarray}
$\ra$ switching  rows  3  and  5,  and  rows  4  and  6 $\ra$
\begin{eqnarray}
\bma
0 & 0 & 0 & 0 \\
1 & a_2 & a_3' & a_4\\
a_1 & a_2' & 1 & a_4'\\	
a_1' & 1 & a_3 & b\\
0 & a_2 & a_3' & 1\\
1 & a_2 & a_3' & a_4'\\
a_1 & 1 & a_3 & a_4\\
a_1' & a_2' & 1 & b'\\	
\ema,
\end{eqnarray}
$\ra$ switch row 1 and 6 $\ra$
\begin{eqnarray}
\label{eq:02}
\bma
1 & a_2 & a_3' & a_4'\\
1 & a_2 & a_3' & a_4\\
a_1 & a_2' & 1 & a_4'\\
a_1' & 1 & a_3 & b\\
0 & a_2 & a_3' & 1\\
0 & 0 & 0 & 0\\
a_1 & 1 & a_3 & a_4\\
a_1' & a_2' & 1 & b'\\
\ema,
\end{eqnarray}
$\ra$ switch column 2 and 3 $\ra$
\begin{eqnarray}
\bma
1 & a_3' & a_2 & a_4'\\
1 & a_3' & a_2 & a_4\\
a_1 & 1 & a_2' & a_4'\\
a_1' & a_3 & 1 & b\\
0 & a_3' & a_2 & 1\\
0 & 0 & 0 & 0\\
a_1 & a_3 & 1 & a_4\\
a_1' & 1 & a_2' & b'\\
\ema
\in \cH_{ACBD},
\end{eqnarray}
$\ra$ switch row 1 and 2, 3 and 7, 4 and 8 $\ra$
\begin{eqnarray}
\bma
1 & a_3' & a_2 & a_4\\
1 & a_3' & a_2 & a_4'\\
a_1 & a_3 & 1 & a_4\\
a_1' & 1 & a_2' & b'\\
0 & a_3' & a_2 & 1\\
0 & 0 & 0 & 0\\
a_1 & 1 & a_2' & a_4'\\
a_1' & a_3 & 1 & b\\
\ema.
\end{eqnarray}
In column 2 and 3 of this matrix, we switch $ (a_3', a_3)$ and $(a_2, a_2')$, in column 4 we rename $(a_4, a_4', b, b')$ as $(a_4', a_4, b', b)$. So we obtain 
\begin{eqnarray}
\label{eq:03}
\bma
1 & a_2 & a_3' & a_4'\\
1 & a_2 & a_3' & a_4\\
a_1 & a_2' & 1 & a_4'\\
a_1' & 1 & a_3 & b\\
0 & a_2 & a_3' & 1\\
0 & 0 & 0 & 0\\
a_1 & 1 & a_3 & a_4\\
a_1' & a_2' & 1 & b'\\
\ema
\in \cH_B \otimes \cH_D \otimes \cH_{AC},
\end{eqnarray}
where the final matrix \eqref{eq:03} is exactly the matrix $\cS_5$ in \eqref{eq:01}. So the set $\cS_6$ in \eqref{eq:01} is a $2\times2\times4$ UPB in $\cH_C \otimes \cH_D \otimes \cH_{AB}$.

In conclusion, we have proven the assertion in terms of the six cases (I)-(VI). This completes the proof.
\end{proof}

\section{Four-partite UPBs from five-qubit UPBs}
\label{sec:res2224}

In this section, we construct a four-partite UPB of size eight, by using an existing five-qubit UPB $\cT$ of size eight in the space $\cH_A \otimes \cH_B \otimes \cH_C \otimes \cH_D \otimes \cH_E$. The UPB $\cT$ was constructed in \cite{Johnston2014The}. For convenience, we describe the matrix of $\cT$ in \eqref{eq:04}. Note that the first row of \eqref{eq:04} means the product state $\ket{0,0,0,0,0}\in \cT$, and one can similarly figure out all elements in $\cT$.  
\begin{eqnarray}
\label{eq:04}
\bma
0 & 0 & 0 & 0 & 0\\
0 & 0 & 1 & a_4 & a_5\\
a_1 & a_2 & a_3 & 1 & a_5'\\
a_1 & a_2 & a_3' & a_4' & 1\\
1 & a_2' & b_3 & b_4 & b_5\\
1 & a_2' & b_3' & c_4 & c_5\\
a_1' & 1 & c_3 & b_4' & c_5'\\
a_1' & 1 & c_3' & c_4' & b_5'\\
\ema.
\end{eqnarray}
We merge two of the systems A, B, C, D and E, so that \eqref{eq:04} corresponds to a set of $2\times2\times2\times4$ orthogonal product vectors. For example, the merge of AB in \eqref{eq:04} implies the set $\{\ket{0,0}\ket{0}\ket{0}\ket{0},...,\ket{a_1',1}\ket{c_3'}\ket{c_4'}\ket{b_5'}\}$. One can verify that there are ten ways for the merge, namely AB, AC, AD, AE, BC, BD, BE, CD, CE and DE. We present the following observation. 

\begin{theorem}
\label{thm:2x2x2x2x2->2x2x2x4}
Suppose \eqref{eq:04} is real. The resulting set by any one of the four merge AB, CD, CE and DE in \eqref{eq:04} is not a $2\times2\times2\times4$ UPB. The resulting set by any one of the six merge AC, AD, AE, BC, BD and BE in \eqref{eq:04} is a $2\times2\times2\times4$ UPB of size eight.
\end{theorem}
\begin{proof}
We prove the claim by ten cases (I)-(X). They respectively work for the system merge AB, CD, CE, DE, AC, AD, AE, BC, BD and BE in \eqref{eq:04}. Using the merge, we shall refer to the set of orthonormal product vectors corresponding to \eqref{eq:04} as $\cT_1,...,\cT_{10}$, respectively, in the ten cases (I)-(X). Further, if any two systems of A, B, C, D, E merge then
evidently they are in $\bbC^4$
. For example, the merge of DE makes the set of four-qubit
orthonormal vectors corresponding to \eqref{eq:04} become $2 \times 2 \times 4$ orthonormal vectors in $\bbC^2\otimes\bbC^2\otimes\bbC^2\otimes\bbC^4=\cH_A\otimes\cH_B\otimes\cH_C\otimes\cH_{DE}$.

(I)	When we merge the system AB, one can verify that the set $\cT_1$ in \eqref{eq:04} is orthogonal to $\ket{1, a_4', a_5, y_4}_{C: D: E: AB} \in \cH_C \otimes \cH_D \otimes \cH_E \otimes \cH_{AB},$ where $\ket{y_4}$ is orthogonal to $\ket{a_1,a_2}$, $\ket{1,a_2'}$ and $\ket{a_1',1}$. So $\cT_1$ is not a UPB in $\cH_C \otimes \cH_D \otimes \cH_E \otimes \cH_{AB}$.

(II) When we merge the system CD, one can verify that the set $\cT_2$ in \eqref{eq:04} is orthogonal to $\ket{a_1', a_2, c_5, y_1}_{A: B: E: CD} \in \cH_A \otimes \cH_B \otimes \cH_E \otimes \cH_{CD}$, where $\ket{y_1}$ is orthogonal to $\ket{0,0}$, $\ket{1,a_4}$ and $\ket{c_3',c_4'}$. So $\cT_2$ is not a UPB in $\cH_A \otimes \cH_B \otimes \cH_E \otimes \cH_{CD}$.

(III) When we merge the system CE, one can verify that the set $\cT_3$ in \eqref{eq:04} is orthogonal to $\ket{a_1', a_2, b_4, y_2}_{A: B: D: CE} \in \cH_A \otimes \cH_B \otimes \cH_D \otimes \cH_{CE}$, where $\ket{y_2}$ is orthogonal to $\ket{0,0}$, $\ket{1,a_5}$ and $\ket{c_3',b_5'}$. So $\cT_3$ is not a UPB in $\cH_A \otimes \cH_B \otimes \cH_D \otimes \cH_{CE}$.

(IV)	When we merge the system DE, one can verify that the set $\cT_4$ in \eqref{eq:04} is orthogonal to $\ket{a_1', a_2, c_3', y_3}_{A: B: C: DE} \in \cH_A \otimes \cH_B \otimes \cH_C \otimes \cH_{DE}$, where $\ket{y_3}$ is orthogonal to $\ket{0,0}$, $\ket{a_4,a_5}$ and $\ket{c_4',b_5'}$. So $\cT_4$ is not a UPB in $\cH_A \otimes \cH_B \otimes \cH_C \otimes \cH_{DE}$.

In the following, we investigate the merge of system AC, AD, AE, BC, BD and BE in \eqref{eq:04}, respectively. For this purpose, we consider the positive numbers $x_j,y_j,w_j \in (0, \pi/2)$ and
\begin{eqnarray}
\label{eq:05}
&&
\ket{ 0 } = \bma 1 \\ 0 \ema,
\quad
\ket{ 1 } = \bma 0 \\ 1 \ema,
\\&&
\ket{ a_1 } = \bma \cos x_1  \\ \sin x_1  \ema,
\quad
\ket{ a_1' } = \bma \sin x_1  \\ -\cos x_1  \ema,
\notag\\&&
\ket{ a_2 } = \bma \cos x_2  \\ \sin x_2  \ema,
\quad
\ket{ a_2' } = \bma \sin x_2  \\ -\cos x_2  \ema,
\notag\\&&
\ket{ a_3 } = \bma \cos x_3  \\ \sin x_3  \ema,
\quad
\ket{ a_3' } = \bma \sin x_3  \\ -\cos x_3  \ema, 
\notag\\&&
\ket{ a_4 } = \bma \cos x_4  \\ \sin x_4  \ema,
\quad
\ket{ a_4' } = \bma \sin x_4  \\ -\cos x_4  \ema,
\notag\\&&
\ket{ a_5 } = \bma \cos x_5  \\ \sin x_5  \ema,
\quad
\ket{ a_5' } = \bma \sin x_5  \\ -\cos x_5  \ema,
\notag\\&&
\ket{ b_3 } = \bma \cos y_3  \\ \sin y_3  \ema,
\quad
\ket{ b_3' } = \bma \sin y_3  \\ -\cos y_3  \ema,
\notag\\&&
\ket{ b_4 } = \bma \cos y_4  \\ \sin y_4  \ema,
\quad
\ket{ b_4' } = \bma \sin y_4  \\ -\cos y_4  \ema,
\notag\\&&
\ket{ b_5 } = \bma \cos y_5  \\ \sin y_5  \ema,
\quad
\ket{ b_5' } = \bma \sin y_5  \\ -\cos y_5  \ema,
\notag\\&&
\ket{ c_3 } = \bma \cos w_3  \\ \sin w_3  \ema,
\quad
\ket{ c_3' } = \bma \sin w_3  \\ -\cos w_3  \ema,
\notag\\&&
\ket{ c_4 } = \bma \cos w_4  \\ \sin w_4  \ema,
\quad
\ket{ c_4' } = \bma \sin w_4  \\ -\cos w_4  \ema,
\notag\\&&
\ket{ c_5 } = \bma \cos w_5  \\ \sin w_5  \ema,
\quad
\ket{ c_5' } = \bma \sin w_5  \\ -\cos w_5  \ema.
\notag
\end{eqnarray}
\begin{eqnarray}
\label{eq:06}
\bma
0 & 0 & 0 & 0 & 0\\
0 & 1 & 0 & a_4 & a_5\\
a_1 & a_3 & a_2 & 1 & a_5'\\
a_1 & a_3' & a_2 & a_4' & 1\\
1 & b_3 & a_2' & b_4 & b_5\\
1 & b_3' & a_2' & c_4 & c_5\\
a_1' & c_3 & 1 & b_4' & c_5'\\
a_1' & c_3' & 1 & c_4' & b_5'\\
\ema.
\end{eqnarray}
(V)	When we merge the system AC of the set $\cT_5$ in \eqref{eq:04}, we switch column B and C, and obtain the matrix in \eqref{eq:06}. The first two columns of \eqref{eq:06} consists of the following eight rows.
\begin{eqnarray}
 &&
 \ket{ e_1 } = \ket{ 0, 0} =  \bma 1  \\ 0  \\ 0  \\ 0 \ema,
 \\&&
 \ket{ e_2 } = \ket{ 0, 1} =  \bma 0  \\ 1  \\ 0  \\ 0 \ema,
 \notag\\&&
 \ket{ e_3 } =  \ket{a_1, a_3} =  \bma  \cos x_1  \cos x_3   \\  \cos x_1  \sin x_3   \\  \sin x_1  \cos x_3   \\  \sin x_1  \sin x_3  \ema,
 \notag\\&&
 \ket{ e_4 } =  \ket{a_1, a_3'} =  \bma  \cos x_1  \sin x_3   \\ - \cos x_1  \cos x_3   \\  \sin x_1  \sin x_3   \\ - \sin x_1   \cos x_3  \ema, 
 \notag\\&&
 \ket{ e_5 } =  \ket{1, b_3} =  \bma 0  \\ 0  \\  \cos y_3   \\  \sin y_3  \ema,
 \notag\\&&
 \ket{ e_6 } = \ket{1, b_3'} =  \bma 0  \\ 0  \\  \sin y_3   \\ - \cos y_3  \ema, 
 \notag\\&&
 \ket{ e_7 } =  \ket{a_1', c_3} = \bma \sin x_1  \cos w_3   \\  \sin x_1  \sin w_3   \\ - \cos x_1  \cos w_3   \\ - \cos x_1  \sin w_3 \ema,
 \notag\\&&
 \ket{ e_8 } =  \ket{a_1', c_3'} =  \bma   \sin x_1  \sin w_3   \\ - \sin x_1  \cos w_3   \\ - \cos x_1  \sin w_3   \\  \cos x_1   \cos w_3   \ema. 
 \notag
\end{eqnarray}

Let $[\ket{ e_1 },..., \ket{ e_8 }]$= 
\begin{eqnarray}
\label{eq:07}
\bma
 1 &  0 &  \cos{x_1}  \cos{x_3} & \cos{x_1}  \sin{x_3} &   0 & 0 & \sin{x_1}  \cos{w_3} &  \sin{x_1}  \sin{w_3}  \\
 0 &  1 &  \cos{x_1}  \sin{x_3} & -\cos{x_1}  \cos{x_3} &  0 & 0 & \sin{x_1}  \sin{w_3} & -\sin{x_1}  \cos{w_3}\\
 0 &  0 &  \sin{x_1}  \cos{x_3} & \sin{x_1}  \sin{x_3} &  \cos{y_3} & \sin{y_3} & -\cos{x_1} \cos{w_3} &  -\cos{x_1}  \sin{w_3}\\
 0 &  0 &  \sin{x_1}  \sin{x_3} & -\sin{x_1}  \cos{x_3} & \sin{y_3} & -\cos{y_3}& -\cos{x_1} \sin{w_3} & \cos{x_1}  \cos{w_3}\\
\ema.
\notag\\
\end{eqnarray}
Suppose that the five-qubit set 
$$\cT_5=\{ \ket{a_j,b_j,c_j,e_j},j=1,...,8\}$$ in \eqref{eq:06}
is orthogonal to the product vector $\ket{u, v, w, y_5}_{B: D: E: AC} \in \cH_B\otimes\cH_D\otimes\cH_E\otimes\cH_{AC}=\bbC^2 \otimes \bbC^2 \otimes \bbC^2 \otimes \bbC^4$. Hence  $\ket{u, v, w}$ is orthogonal to some of the eight three-qubit product vectors $\ket{a_1,b_1,c_1},..., \ket{a_8,b_8,c_8}$. Up to the permutation of subscripts, we may assume that  $\ket{u, v, w}$ is orthogonal to k product vectors $\ket{a_1,b_1,c_1},..., \ket{a_k,b_k,c_k}$, and  $\ket{y_5}$ is orthogonal to (8-k) product vectors  $\ket{f_{k+1}},..., \ket{f_8}$.
By checking the expression of $\cT_5$, one can show the property that each column of the matrix of $\cT_5$ has at most two identical elements. In particular, the set $\{\ket{a_j}\}$ has at most two identical elements, and the set $\{\ket{b_j}\}$, $\{\ket{c_j}\}$ has no identical elements. So we obtain $k<5$, and it suffices to study the case k=4. Hence,  $\ket{u, v, w}$ is orthogonal to exactly four of the eight product vectors  $\ket{a_1,b_1,c_1},..., \ket{a_8,b_8,c_8}$. 
At the same time,  $\ket{y_5}$ is orthogonal to four of the eight product vectors  $\ket{e_1},...,  \ket{e_8}$. So the four vectors form a $4\times4$ singular matrix, which has rank at most three. However, by applying Matlab to matrix \eqref{eq:07}, we have shown that such a matrix does not exist.  Hence, the product vector  $\ket{u, v, w, y_5}_{B: D: E: AC}$ does not exist. We have proven that $\cT_5=\{\ket{a_j,b_j,c_j,e_j}, j=1,...,8\}$ is a four-partite UPB in $\cH_B\otimes\cH_D\otimes\cH_E\otimes\cH_{AC}$.

(VI)	When we merge the system AD of the set $\cT_6$ in \eqref{eq:04}, we switch column D $\rightarrow$ B $\rightarrow$ C $\rightarrow$ D, and obtain the matrix in \eqref{eq:08}.
\begin{eqnarray}
\label{eq:08}
\bma
 0 & 0 & 0 & 0 & 0\\ 
 0 & a_4 & 0 & 1 & a_5\\ 
 a_1 & 1 & a_2 & a_3 & a_5'\\ 
 a_1 & a_4' & a_2 & a_3' & 1\\
 1 & b_4 & a_2' & b_3 & b_5\\ 
 1 & c_4 & a_2' & b_3' & c_5\\
 a_1' & b_4' & 1 & c_3 & c_5'\\
 a_1' & c_4' & 1 & c_3' & b_5'\\
\ema.
\end{eqnarray}

The  first  two  columns  of \eqref{eq:08} consists  of  the following eight rows.
\begin{eqnarray}
&&
 \ket{ f_1}   = \ket{ 0, 0 } = \bma 1 \\ 0 \\ 0 \\ 0 \ema,
 \\&&
 \ket{ f_2 }  = \ket{ 0, a_4 } = \bma \cos x_4  \\ \sin x_4  \\ 0 \\ 0 \ema,
 \notag\\&&
 \ket{ f_3 }  =  \ket{a_1, 1 } = \bma 0\\ \cos x_1  \\ 0 \\ \sin x_1  \ema,
 \notag\\&&
 \ket{ f_4 }  =  \ket{a_1, a_4' } = \bma \cos x_1  \sin x_4  \\ -\cos x_1  \cos x_4  \\ \sin x_1  \sin x_4  \\ -\sin x_1   \cos x_4 \ema,
 \notag\\&&
 \ket{ f_5 }  =  \ket{1, b_4}  = \bma 0 \\ 0 \\ \cos y_4  \\ \sin y_4 \ema, 
 \notag\\&&
 \ket{ f_6 }  = \ket{1, c_4}  = \bma 0 \\ 0 \\ \cos w_4  \\ \sin w_4 \ema,
 \notag\\&&
 \ket{ f_7 }  =  \ket{a_1', b_4'}  = \bma  \sin x_1  \sin y_4  \\ -\sin x_1  \cos y_4  \\ -\cos x_1  \sin y_4  \\ \cos x_1   \cos y_4 \ema,
 \notag\\&&
 \ket{ f_8}   =  \ket{a_1', c_4'}  =  \bma \sin x_1  \sin w_4  \\ -\sin x_1  \cos w_4  \\ -\cos x_1  \sin w_4  \\ \cos x_1   \cos w_4 \ema.
 \notag
\end{eqnarray}

 Let   $[\ket{ f_1 }, ...,  \ket{ f_8 }]$  =
 \begin{eqnarray}
 \label{eq:09}
 \bma
 1 & \cos x_4 & 0 & \cos x_1  \sin x_4 & 0 & 0 & \sin x_1  \sin y_4 & \sin x_1  \sin w_4  \\
 0 & \sin x_4 & \cos x_1 & -\cos x_1  \cos x_4 & 0 & 0 & -\sin x_1  \cos y_4 &  -\sin x_1  \cos w_4 \\
 0 & 0 & 0 & \sin x_1  \sin x_4 & \cos y_4 & \cos w_4 & -\cos x_1  \sin y_4 &  -\cos x_1  \sin w_4 \\
 0 & 0 & \sin x_1 & -\sin x_1  \cos x_4 & \sin y_4 & \sin w_4 & \cos x_1  \cos y_4 &  \cos x_1  \cos w_4 \\
 \ema
\end{eqnarray}
Suppose that the five-qubit set 
$$\cT_6=\{ \ket{a_j,b_j,c_j,f_j},j=1,...,8\}$$ in \eqref{eq:08}
is orthogonal to the product vector $\ket{u, v, w, y_6}_{B: C: E: AD} \in \bbC^2 \otimes \bbC^2 \otimes \bbC^2 \otimes \bbC^4$. Hence  $\ket{u, v, w}$ is orthogonal to some of the eight three-qubit product vectors $\ket{a_1,b_1,c_1},..., \ket{a_8,b_8,c_8}$. Up to the permutation of subscripts, we may assume that  $\ket{u, v, w}$ is orthogonal to k product vectors $\ket{a_1,b_1,c_1},..., \ket{a_k,b_k,c_k}$, and  $\ket{y_6}$ is orthogonal to (8-k) product vectors  $\ket{f_{k+1}},..., \ket{f_8}$.
By checking the expression of $\cT_6$, one can show the property that each column of the matrix of $\cT_6$ has at most two identical elements. In particular, the set $\{\ket{a_j}\}$ has at most two identical elements, and the set $\{\ket{b_j}\}$, $\{\ket{c_j}\}$ has no  identical elements. So we obtain $k<5$, and it suffices to study the case $k=4$. Hence,  $\ket{u, v, w}$ is orthogonal to exactly four of the eight product vectors  $\ket{a_1,b_1,c_1},..., \ket{a_8,b_8,c_8}$. 
At the same time,  $\ket{y_6}$ is orthogonal to four of the eight product vectors  $\ket{f_1},...,  \ket{f_8}$. So the four vectors form a $4\times4$ singular matrix, which has rank at most three. However, by applying Matlab to matrix \eqref{eq:09}, we have shown that such a matrix does not exist.  Hence, the product vector  $\ket{u, v, w, y_6}_{B: C: E: AD}$ does not exist. We have proven that $\cT_6=\{\ket{a_j,b_j,c_j,f_j}, j=1,...,8\}$ is a four-partite UPB in $\cH_{B: C: E: AD} = \bbC^2 \otimes \bbC^2 \otimes \bbC^2 \otimes \bbC^4$.

(VII) When we merge the system AD of the set $\cT_7$ in \eqref{eq:04}, we switch column E $\rightarrow$ B $\rightarrow$ C $\rightarrow$ D $\rightarrow$ E, and obtain the matrix in \eqref{eq:10}. 
\begin{eqnarray}
\label{eq:10}
\bma
 0 & 0 & 0 & 0 & 0\\ 
 0 & a_5 & 0 & 1 & a_4\\ 
 a_1 & a_5' & a_2 & a_3 & 1\\ 
 a_1 & 1 & a_2 & a_3' & a_4'\\
 1 & b_5 & a_2' & b_3 & b_4\\ 
 1 & c_5 & a_2' & b_3' & c_4\\
 a_1' & c_5' & 1 & c_3 & b_4'\\
 a_1' & b_5' & 1 & c_3' & c_4'\\
\ema.
\end{eqnarray}
The  first  two  columns  of \eqref{eq:10} consists  of  the following eight rows.
\begin{eqnarray}
&&
 \ket{  g_1 } = \ket{ 0, 0} =  \bma 1 \\ 0 \\ 0 \\ 0 \ema, 
 \\&&
 \ket{  g_2 } = \ket{ 0, a_5} = \bma \cos x_5  \\ \sin x_5 \\ 0 \\ 0 \ema, 
 \notag\\&&
 \ket{  g_3 } = \ket{a_1, a_5'} = \bma \cos x_1  \sin x_5  \\ -\cos x_1  \cos x_5  \\ \sin x_1  \sin x_5  \\ -\sin x_1   \cos x_5   \ema, 
  \notag\\&&
 \ket{  g_4 } = \ket{a_1, 1} = \bma 0 \\ \cos x_1  \\ 0 \\ \sin x_1  \ema,
  \notag\\&&
 \ket{  g_5 } =  \ket{1, b_5} = \bma 0 \\ 0 \\ \cos y_5  \\ \sin y_5  \ema, 
  \notag\\&&
 \ket{  g_6 } = \ket{1, c_5} = \bma 0 \\ 0 \\ \cos w_5  \\ \sin w_5  \ema, 
  \notag\\&&
 \ket{  g_7 } = \ket{a_1', c_5'} = \bma  \sin x_1  \sin w_5  \\ -\sin x_1  \cos w_5  \\ -\cos x_1  \sin w_5  \\ \cos x_1   \cos w_5    \ema, 
  \notag\\&&
 \ket{  g_8 } = \ket{a_1', b_5'} = \bma  \sin x_1  \sin y_5  \\ -\sin x_1  \cos y_5  \\ -\cos x_1  \sin y_5  \\ \cos x_1   \cos y_5    \ema.
  \notag
\end{eqnarray}

Let   $\ket{  g_1 }$, ...,  $\ket{  g_8 }$  =
\begin{eqnarray}
\label{eq:11}
\bma
 1 & \cos x_5 &  \cos x_1  \sin x_5 & 0 & 0 & 0 & \sin x_1  \sin w_5 & \sin x_1  \sin y_5\\
 0 & \sin x_5 & -\cos x_1  \cos x_5 & \cos x_1 & 0 & 0 & -\sin x_1  \cos w_5 & -\sin x_1  \cos y_5\\
0 & 0 & \sin x_1  \sin x_5 & 0 & \cos y_5 & \cos w_5 & -\cos x_1  \sin w_5 & -\cos x_1  \sin y_5\\
 0 & 0 & -\sin x_1  \cos x_5 & \sin x_1 & \sin y_5 & \sin w_5 & \cos x_1  \cos w_5 &  \cos  x_1 \cos y_5\\
 \ema.
\end{eqnarray}

Suppose that the five-qubit set 
$$\cT_7=\{ \ket{a_j,b_j,c_j,g_j},j=1,...,8\}$$ in \eqref{eq:10}
is orthogonal to the product vector $\ket{u, v, w, y_7}_{B: C: D: AE} \in \bbC^2 \otimes \bbC^2 \otimes \bbC^2 \otimes \bbC^4$. Hence  $\ket{u, v, w}$ is orthogonal to some of the eight three-qubit product vectors $\ket{a_1,b_1,c_1},..., \ket{a_8,b_8,c_8}$. Up to the permutation of subscripts, we may assume that  $\ket{u, v, w}$ is orthogonal to k product vectors $\ket{a_1,b_1,c_1},..., \ket{a_k,b_k,c_k}$, and  $\ket{y_7}$ is orthogonal to (8-k) product vectors  $\ket{g_{k+1}},..., \ket{g_8}$.
By checking the expression of $\cT_7$, one can show the property that each column of the matrix of $\cT_7$ has at most two identical elements. In particular, the set $\{\ket{a_j}\}$ has at most two identical elements, and the set $\{\ket{b_j}\}$, $\{\ket{c_j}\}$ has no  identical elements. So we obtain $k<5$, and it suffices to study the case $k=4$. Hence,  $\ket{u, v, w}$ is orthogonal to exactly four of the eight product vectors  $\ket{a_1,b_1,c_1},..., \ket{a_8,b_8,c_8}$. 
At the same time,  $\ket{y_7}$ is orthogonal to four of the eight product vectors  $\ket{g_1},...,  \ket{g_8}$. So the four vectors form a $4\times4$ singular matrix, which has rank at most three. However, by applying Matlab to matrix \eqref{eq:11}, we have shown that such a matrix does not exist.  Hence, the product vector  $\ket{u, v, w, y_7}_{B: D: E: AC}$ does not exist. We have proven that $\cT_7=\{\ket{a_j,b_j,c_j,g_j}, j=1,...,8\}$ is a four-partite UPB in $\cH_{B: C: D: AE} = \bbC^2 \otimes \bbC^2 \otimes \bbC^2 \otimes \bbC^4$.

(VIII)	When we merge the system BC of the set $\cT_8$ in \eqref{eq:04}, we switch column $C \rightarrow B \rightarrow A \rightarrow C$, and obtain the matrix in \eqref{eq:12}.
\begin{eqnarray}
\label{eq:12}
\bma
0 & 0 & 0 & 0 & 0\\
0 & 1 & 0 & a_4 & a_5\\
a_2 & a_3 & a_1 & 1 & a_5'\\
a_2 & a_3' & a_1 & a_4' & 1\\
a_2' & b_3 & 1 & b_4 & b_5\\
a_2' & b_3' & 1 & c_4 & c_5\\
1 & c_3 & a_1' & b_4' & c_5'\\
1 & c_3' & a_1' & c_4' & b_5'\\
\ema.
\end{eqnarray}
The first two columns of \eqref{eq:12} consists of the following eight rows.
\begin{eqnarray}
 &&
 \ket{ h_1 } = \ket{ 0, 0} =  \bma 1  \\ 0  \\ 0  \\ 0 \ema,
 \\&&
 \ket{ h_2 } = \ket{ 0, 1} =  \bma 0  \\ 1  \\ 0  \\ 0 \ema,
 \notag\\&&
 \ket{ h_3 } =  \ket{a_2, a_3} =  \bma  \cos x_2  \cos x_3   \\  \cos x_2  \sin x_3   \\  \sin x_2  \cos x_3   \\  \sin x_2  \sin x_3  \ema,
 \notag\\&&
 \ket{ h_4 } =  \ket{a_2, a_3'} =  \bma  \cos x_2  \sin x_3   \\ - \cos x_2  \cos x_3   \\  \sin x_2  \sin x_3   \\ - \sin x_2   \cos x_3  \ema, 
 \notag\\&&
 \ket{ h_5 } =  \ket{a_2', b_3} =  \bma  \sin x_2 \cos y_3 \\  \sin x_2 \sin y_3 \\  -\cos x_2 \cos y_3   \\  -\cos x_2 \sin y_3  \ema,
 \notag\\&&
 \ket{ h_6 } = \ket{a_2', b_3'} =  \bma \sin x_2 \sin y_3  \\ -\sin x_2 \cos y_3  \\  -\cos x_2 \sin y_3   \\ \cos x_2 \cos y_3  \ema, 
 \notag\\&&
 \ket{ h_7 } =  \ket{1, c_3} = \bma 0 \\  0  \\ \cos w_3  \\  \sin w_3 \ema,
 \notag\\&&
 \ket{ h_8 } =  \ket{1, c_3'} =  \bma  0 \\  0  \\ \sin w_3   \\  -\cos w_3  \ema. 
 \notag
\end{eqnarray}

Let $[\ket{ h_1 },..., \ket{ h_8 }]$= 
\begin{eqnarray}
\label{eq:13}
\bma
1 & 0 & \cos{x_2} \cos{x_3} & \cos{x_2} \sin{x_3} & \sin{x_2} \cos{y_3} & \sin{x_2} \sin{y_3} &  0 & 0 \\
0 & 1 & \cos{x_2} \sin{x_3} & -\cos{x_2} \cos{x_3} & \sin{x_2} \sin{y_3} & -\sin{x_2} \cos{y_3} & 0 & 0  \\
0 & 0 & \sin{x_2} \cos{x_3} & \sin{x_2} \sin{x_3} & -\cos{x_2} \cos{y_3} & -\cos{x_2} \sin{y_3} & \cos{w_3} & \sin{w_3}\\
0 & 0 & \sin{x_2} \sin{x_3} & -\sin{x_2} \cos{x_3} & -\cos{x_2} \sin{y_3} & \cos {x_2} \cos{y_3} & \sin{w_3} & -\cos{w_3}\\
\ema.
\notag\\
\end{eqnarray}
Suppose that the five-qubit set 
$$\cT_8=\{ \ket{a_j,b_j,c_j,h_j},j=1,...,8\}$$
in \eqref{eq:12} is orthogonal to the product vector $\ket{u, v, w, y_8}_{A: D: E: BC} \in \bbC^2 \otimes \bbC^2 \otimes \bbC^2 \otimes \bbC^4$. Hence  $\ket{u, v, w}$ is orthogonal to some of the eight three-qubit product vectors $\ket{a_1,b_1,c_1},..., \ket{a_8,b_8,c_8}$. Up to the permutation of subscripts, we may assume that  $\ket{u, v, w}$ is orthogonal to k product vectors $\ket{a_1,b_1,c_1},..., \ket{a_k,b_k,c_k}$, and  $\ket{y_8}$ is orthogonal to (8-k) product vectors  $\ket{h_{k+1}},..., \ket{h_8}$.
By checking the expression of $\cT_8$, one can show the property that each column of the matrix of $\cT_8$ has at most two identical elements. In particular, the set $\{\ket{a_j}\}$ has at most two identical elements, and the set $\{\ket{b_j}\}$, $\{\ket{c_j}\}$ has no  identical elements. So we obtain $k<5$, and it suffices to study the case $k=4$. Hence,  $\ket{u, v, w}$ is orthogonal to exactly four of the eight product vectors  $\ket{a_1,b_1,c_1},..., \ket{a_8,b_8,c_8}$. 
At the same time,  $\ket{y_8}$ is orthogonal to four of the eight product vectors  $\ket{h_1},...,  \ket{h_8}$. So the four vectors form a $4\times4$ singular matrix, which has rank at most three. However, by applying Matlab to matrix \eqref{eq:13}, we have shown that such a matrix does not exist.  Hence, the product vector  $\ket{u, v, w, y_8}_{A: D: E: BC}$ does not exist. We have proven that $\cT_8=\{\ket{a_j,b_j,c_j,h_j}, j=1,...,8\}$ is a four-partite UPB in $\cH_{A: D: E: BC} = \bbC^2 \otimes \bbC^2 \otimes \bbC^2 \otimes \bbC^4$.

(IX)	When we merge the system BD of the set $\cT_9$ in \eqref{eq:04}, we switch column $D \rightarrow B \rightarrow A \rightarrow C \rightarrow D$, and obtain the matrix in \eqref{eq:14}.
\begin{eqnarray}
\label{eq:14}
\bma
0 & 0 & 0 & 0 & 0\\
0 & a_4 & 0 & 1 & a_5\\
a_2 & 1 & a_1 & a_3 & a_5'\\
a_2 & a_4' & a_1 & a_3' & 1\\
a_2' & b_4 & 1 & b_3 & b_5\\
a_2' & c_4 & 1 & b_3' & c_5\\
1 & b_4' & a_1' & c_3 & c_5'\\
1 & c_4' & a_1' & c_3' & b_5'\\
\ema.
\end{eqnarray}
The first two columns of \eqref{eq:14} consists of the following eight rows.
\begin{eqnarray}
 &&
 \ket{ j_1 } = \ket{ 0, 0} =  \bma 1  \\ 0  \\ 0  \\ 0 \ema,
 \\&&
 \ket{ j_2 } = \ket{ 0, a_4} =  \bma \cos x_4  \\ \sin x_4  \\ 0  \\ 0 \ema,
 \notag\\&&
 \ket{ j_3 } =  \ket{a_2, 1} =  \bma  0  \\  \cos x_2    \\ 0 \\  \sin x_2  \ema,
 \notag\\&&
 \ket{ j_4 } =  \ket{a_2, a_4'} =  \bma  \cos x_2  \sin x_4   \\ - \cos x_2  \cos x_4   \\  \sin x_2  \sin x_4   \\ - \sin x_2   \cos x_4  \ema, 
 \notag\\&&
 \ket{ j_5 } =  \ket{a_2', b_4} =  \bma  \sin x_2 \cos y_4 \\  \sin x_2 \sin y_4 \\  -\cos x_2 \cos y_4   \\  -\cos x_2 \sin y_4  \ema,
 \notag\\&&
 \ket{ j_6 } = \ket{a_2', c_4} =  \bma \sin x_2 \cos w_4  \\ \sin x_2 \sin w_4  \\  -\cos x_2 \cos w_4   \\ -\cos x_2 \sin w_4  \ema, 
 \notag\\&&
 \ket{ j_7 } =  \ket{1, b_4'} = \bma 0 \\  0  \\ \sin y_4  \\  -\cos y_4 \ema,
 \notag\\&&
 \ket{ j_8 } =  \ket{1, c_4'} =  \bma  0 \\  0  \\ \sin w_4   \\  -\cos w_4  \ema. 
 \notag
\end{eqnarray}

Let $[\ket{ j_1 },..., \ket{ j_8 }]$= 
\begin{eqnarray}
\label{eq:15}
\bma
1 & \cos x_4 & 0 & \cos{x_2} \sin{x_4} & \sin{x_2} \cos{y_4} & \sin{x_2} \cos{w_4} &  0 & 0 \\
0 & \sin x_4 & \cos{x_2} & -\cos{x_2} \cos{x_4} & \sin{x_2} \sin{y_4} & \sin{x_2} \sin{w_4} & 0 & 0  \\
0 & 0 & 0 & \sin{x_2} \sin{x_4} & -\cos{x_2} \cos{y_4} & -\cos{x_2} \cos{w_4} & \sin{y_4} & \sin{w_4}\\
0 & 0 & \sin{x_2} & -\sin{x_2} \cos{x_4} & -\cos{x_2} \sin{y_4} & -\cos {x_2} \sin{w_4} & -\cos{y_4} & -\cos{w_4}\\
\ema.
\notag\\
\end{eqnarray}
Suppose that the five-qubit set 
$$\cT_9=\{ \ket{a_j,b_j,c_j,m_j},j=1,...,8\}$$
in \eqref{eq:14} is orthogonal to the product vector $\ket{u, v, w, y_9}_{A: C: E: BD} \in \bbC^2 \otimes \bbC^2 \otimes \bbC^2 \otimes \bbC^4$. Hence  $\ket{u, v, w}$ is orthogonal to some of the eight three-qubit product vectors $\ket{a_1,b_1,c_1},..., \ket{a_8,b_8,c_8}$. Up to the permutation of subscripts, we may assume that  $\ket{u, v, w}$ is orthogonal to k product vectors $\ket{a_1,b_1,c_1},..., \ket{a_k,b_k,c_k}$, and  $\ket{y_9}$ is orthogonal to (8-k) product vectors  $\ket{m_{k+1}},..., \ket{m_8}$.
By checking the expression of $\cT_9$, one can show the property that each column of the matrix of $\cT_9$ has at most two identical elements. In particular, the set $\{\ket{a_j}\}$ has at most two identical elements, and the set $\{\ket{b_j}\}$, $\{\ket{c_j}\}$ has no  identical elements. So we obtain $k<5$, and it suffices to study the case $k=4$. Hence,  $\ket{u, v, w}$ is orthogonal to exactly four of the eight product vectors  $\ket{a_1,b_1,c_1},..., \ket{a_8,b_8,c_8}$. 
At the same time,  $\ket{y_9}$ is orthogonal to four of the eight product vectors  $\ket{m_1},...,  \ket{m_8}$. So the four vectors form a $4\times4$ singular matrix, which has rank at most three. However, by applying Matlab to matrix \eqref{eq:15}, we have shown that such a matrix does not exist.  Hence, the product vector  $\ket{u, v, w, y_9}_{A: C: E: BD}$ does not exist. We have proven that $\cT_9=\{\ket{a_j,b_j,c_j,m_j}, j=1,...,8\}$ is a four-partite UPB in $\cH_{A: C: E: BD} = \bbC^2 \otimes \bbC^2 \otimes \bbC^2 \otimes \bbC^4$.

(X)	When we merge the system BE of the set $\cT_{10}$ in \eqref{eq:04}, we switch column $E \rightarrow B \rightarrow A \rightarrow C \rightarrow D \rightarrow E$ of the set $\cT$, and obtain the matrix in \eqref{eq:16}.
\begin{eqnarray}
\label{eq:16}
\bma
0 & 0 & 0 & 0 & 0\\
0 & a_5 & 0 & 1 & a_4\\
a_2 & a_5' & a_1 & a_3 & 1\\
a_2 & 1 & a_1 & a_3' & a_4'\\
a_2' & b_5 & 1 & b_3 & b_4\\
a_2' & c_5 & 1 & b_3' & c_4\\
1 & c_5' & a_1' & c_3 & b_4'\\
1 & b_5' & a_1' & c_3' & c_4'\\
\ema.
\end{eqnarray}
The first two columns of \eqref{eq:16} consists of the following eight rows.
\begin{eqnarray}
 &&
 \ket{ l_1 } = \ket{ 0, 0} =  \bma 1  \\ 0  \\ 0  \\ 0 \ema,
 \\&&
 \ket{ l_2 } = \ket{ 0, a_5} =  \bma \cos x_5  \\ \sin x_5  \\ 0  \\ 0 \ema,
 \notag\\&&
 \ket{ l_3 } =  \ket{a_2, a_5'} =  \bma  \cos x_2 \sin x_5  \\  -\cos x_2 \cos x_5   \\ \sin x_2 \sin x_5 \\  -\sin x_2 \cos x_5  \ema,
 \notag\\&&
 \ket{ l_4 } =  \ket{a_2, 1} =  \bma  0   \\ \cos x_2   \\  0  \\ \sin x_2 \ema,
 \notag\\&&
 \ket{ l_5 } =  \ket{a_2', b_5} =  \bma  \sin x_2 \cos y_5 \\  \sin x_2 \sin y_5 \\  -\cos x_2 \cos y_5   \\  -\cos x_2 \sin y_5  \ema,
 \notag\\&&
 \ket{ l_6 } = \ket{a_2', c_5} =  \bma \sin x_2 \cos w_5  \\ \sin x_2 \sin w_5  \\  -\cos x_2 \cos w_5   \\ -\cos x_2 \sin w_5  \ema, 
 \notag\\&&
 \ket{ l_7 } =  \ket{1, c_5'} = \bma 0 \\  0  \\ \sin w_5  \\  -\cos w_5 \ema,
 \notag\\&& 
 \ket{ l_8 } =  \ket{1, b_5'} =  \bma  0 \\  0  \\ \sin y_5   \\  -\cos y_5  \ema. 
 \notag
\end{eqnarray}
Let $[\ket{ l_1 },..., \ket{ l_8 }]$= 
\begin{eqnarray}
\label{eq:17}
\bma
1 & \cos x_5 & \cos{x_2} \sin{x_5} & 0 & \sin{x_2} \cos{y_5} & \sin{x_2} \cos{w_5} &  0 & 0 \\
0 & \sin x_5 & -\cos{x_2} \cos{x_5} & \cos x_2 & \sin{x_2} \sin{y_5} & \sin{x_2} \sin{w_5} & 0 & 0  \\
0 & 0 &  \sin{x_2} \sin{x_5} & 0 & -\cos{x_2} \cos{y_5} & -\cos{x_2} \cos{w_5} & \sin{w_5} & \sin{y_5}\\
0 & 0 &  -\sin{x_2} \cos{x_5} & \sin x_2 & -\cos{x_2} \sin{y_5} & -\cos {x_2} \sin{w_5} & -\cos{w_5} & -\cos{y_5}\\
\ema.
\notag\\
\end{eqnarray}
Suppose that the five-qubit set 
$$\cT_{10}=\{ \ket{a_j,b_j,c_j,n_j},j=1,...,8\}$$ in \eqref{eq:16}
is orthogonal to the product vector $\ket{u, v, w, y_{10}}_{A: C: D: BE} \in \bbC^2 \otimes \bbC^2 \otimes \bbC^2 \otimes \bbC^4$. Hence  $\ket{u, v, w}$ is orthogonal to some of the eight three-qubit product vectors $\ket{a_1,b_1,c_1},..., \ket{a_8,b_8,c_8}$. Up to the permutation of subscripts, we may assume that  $\ket{u, v, w}$ is orthogonal to k product vectors $\ket{a_1,b_1,c_1},..., \ket{a_k,b_k,c_k}$, and  $\ket{y_{10}}$ is orthogonal to (8-k) product vectors  $\ket{n_{k+1}},..., \ket{n_8}$.
By checking the expression of $\cT_{10}$, one can show the property that each column of the matrix of $\cT_{10}$ has at most two identical elements. In particular, the set $\{\ket{a_j}\}$ has at most two identical elements, and the set $\{\ket{b_j}\}$, $\{\ket{c_j}$\} has no  identical elements. So we obtain $k<5$, and it suffices to study the case $k=4$. Hence,  $\ket{u, v, w}$ is orthogonal to exactly four of the eight product vectors  $\ket{a_1,b_1,c_1},..., \ket{a_8,b_8,c_8}$. 
At the same time,  $\ket{y_{10}}$ is orthogonal to four of the eight product vectors  $\ket{n_1},...,  \ket{n_8}$. So the four vectors form a $4\times4$ singular matrix, which has rank at most three. However, by applying Matlab to matrix \eqref{eq:17}, we have shown that such a matrix does not exist.  Hence, the product vector  $\ket{u, v, w, y_{10}}_{A: C: D: BE}$ does not exist. We have proven that $\cT_{10}=\{\ket{a_j,b_j,c_j,n_j}, j=1,...,8\}$ is a four-partite UPB in $\cH_{A: C: D: BE} = \bbC^2 \otimes \bbC^2 \otimes \bbC^2 \otimes \bbC^4$.
\end{proof}

So we have managed to construct a family of real four-partite UPBs of size eight, by using the existing five-qubit UPBs. It remains to show whether the real UPB in Theorem \ref{thm:2x2x2x2x2->2x2x2x4} can be extended to complex field. Our primary investigation shows that it might be true by using the similar technique in the above proof.

\section{Application}
\label{sec:app}

In this section, we apply the results in the last two sections. We shall construct two families of multipartite entangled states, which have positive partial transpose (PPT). A bipartite state $\r\in\cB(\bbC^m\otimes\bbC^n):=\cB(\cH_A\otimes\cH_B)$ is a positive semidefinite matrix of trace one. The partial transpose of $\r$ is defined as $\r^\G=\sum_{j,k}(\ketbra{j}{k}\otimes I)\r(\ketbra{j}{k}\otimes I)$ with the computational basis $\{\ket{j}\}\in\bbC^m$. We say that $\r$ is a PPT state when $\r^\G$ is also positive semidefinite. For example, the separable state is PPT, because it is the convex sum of product states. On the other hand, the PPT state may be not separable when $\min\{m,n\}>2$ \cite{peres1996,hhh96}. 
The examples of
two-qutrit PPT entangled states were constructed in 1980s, and then introduced in quantum information
\cite{horodecki1997,MD1975Completely,E1982Decomposable}. 
Next, two-qutrit PPT entangled states of rank four were fully characterized in
\cite{cd11JMP} and \cite{skowronek11}, respectively. In contrast, the construction of multipartite PPT entangled states and their investigation in terms of entanglement theory is more involved. A systematic method for such a construction is to employ an $n$-partite UPB $\{\ket{u_j}\}_{j=1,2,...,m}\in\bbC^{d_1}\otimes...\otimes\bbC^{d_n}$. That is, one can show that
\begin{eqnarray}
\label{eq:rho}
\rho=\frac{1}{n}\bigg(
I-\sum^m_{j=1}\proj{u_j}
\bigg),
\end{eqnarray}
is an $n$-partite PPT entangled state of rank $d_1...d_n-m$. In Theorems \ref{thm:2x2x2x2->2x2x4} and \ref{thm:2x2x2x2x2->2x2x2x4}, we have constructed 
respectively two and six UPBs. Every one of them can be applied to \eqref{eq:rho} and generate a PPT entangled state $\rho$. For example, we demonstrate a $4\times2\times2$ PPT entangled state of rank eight from Theorem \ref{thm:2x2x2x2->2x2x4} as follows.
\begin{eqnarray}
\label{eq:rho1}
\rho=&&
\frac{1}{8}\bigg(
I-\proj{0,0,0,0}
-\proj{1,a_2,a_3',a_4}
-\proj{0,a_2,a_3',1}
\notag\\-&&
\proj{1,a_2,a_3',a_4'}
-\proj{a_1,a_2',1,a_4'}
-\proj{a_1',1,a_3,b_4}
\notag\\-&&
\proj{a_1,1,a_3,a_4}
-\proj{a_1',a_2',1,b_4'}
\bigg)
\notag\\:=&& 
\frac{1}{8}\bigg(
I-P\bigg),
\end{eqnarray}
where system $A,B$ are merged.
Further, we assume that $\a$ is the four-qubit PPT entangled state having the same expression as that of $\r$. 

We investigate the geometric measure of entanglement of states $\r$ and $\a$ \cite{wg2003,cxz2010}. For any $n$-partite quantum state $\s$, the measure is defined as
\begin{eqnarray}
\label{eq:gme}
G(\s):=
-\log_2
\max_{\d_1,...,\d_n}
\bra{\d_1,...,\d_n}	
\s
\ket{\d_1,...,\d_n},
\end{eqnarray}
where $\ket{\d_1,...,\d_n}$ is a normalized product state in $\bbC^{d_1}\otimes...\otimes\bbC^{d_n}$. In particular, $\r\in\cB(\bbC^4\otimes\bbC^2\otimes\bbC^2)$ means that $d_1=4$ and $d_2=d_3=2$ with $n=3$, and $\a\in\cB(\bbC^2\otimes\bbC^2\otimes\bbC^2\otimes\bbC^2)$ means that $d_1=d_2=d_3=d_4=2$ with $n=4$. Hence we obtain that
\begin{eqnarray}
G(\r)\le G(\a).
\end{eqnarray}
So $G(\r)$ is upper bounded by $G(\a)$. However it is not easy to compute, and also might not be a tight upper bound. 
Eqs. \eqref{eq:rho1} and \eqref{eq:gme} imply that, the evaluation of $G(\r)$ is equivalent to
\begin{eqnarray}
\label{min:a1nPa1n}
\min_{\ket{\d_1,\d_2,\d_3}\in\bbC^4\otimes\bbC^2\otimes\bbC^2}
\bra{\d_1,\d_2,\d_3}	
P
\ket{\d_1,\d_2,\d_3},
\end{eqnarray}
for $P\in\cB(\bbC^4\otimes\bbC^2\otimes\bbC^2)$ in \eqref{eq:rho1}. To find an upper bound, we only consider $\ket{\d_1,\d_2,\d_3}\in\bbR^4\otimes\bbR^2\otimes\bbR^2$ in \eqref{min:a1nPa1n}. We assume that
\begin{eqnarray}
\label{eq:a1}
\ket{\d_1}
=
\bma
\cos\n_1\cos\m_1
\\
\cos\n_1\sin\m_1
\\
\sin\n_1\cos\m_2
\\
\sin\n_1\sin\m_2
\ema,	
\quad
\ket{\d_2}
=
\bma
\cos\n_2
\\
\sin\n_2
\ema,	
\quad 
\ket{\d_3}
=
\bma
\cos\n_3
\\
\sin\n_3
\ema,
\end{eqnarray}
where $\n_j,\m_j\in[0,2\p)$. So \eqref{min:a1nPa1n} is upper bounded by 
\begin{eqnarray}
\label{min:a1nPa1n-1}
\min_{\ket{\d_1,\d_2,\d_3}\in\bbR^4\otimes\bbR^2\otimes\bbR^2}
\bra{\d_1,\d_2,\d_3}	
P
\ket{\d_1,\d_2,\d_3}.
\end{eqnarray}
Using Eqs. \eqref{eq:ket0} -\eqref{eq:c8}, \eqref{eq:rho1} and \eqref{eq:a1}, we can write up the target function in \eqref{min:a1nPa1n-1} as
\begin{eqnarray}
\label{eq:d1d2d3}
&&
\bra{\d_1,\d_2,\d_3}\bigg(
\proj{0,0,0,0}
+\proj{1,a_2,a_3',a_4}
+\proj{0,a_2,a_3',1}
\notag\\+&&
\proj{1,a_2,a_3',a_4'}
+\proj{a_1,a_2',1,a_4'}
+\proj{a_1',1,a_3,b_4}
\notag\\+&&
\proj{a_1,1,a_3,a_4}
+\proj{a_1',a_2',1,b_4'}
\bigg)
\ket{\d_1,\d_2,\d_3}
\notag\\=&&
(\cos\n_1\cos\m_1
\cos\n_2
\cos\n_3)^2
\notag\\+&&
\bigg(
(\cos x_2 \sin \n_1 \cos \m_2 + \sin x_2 \sin \n_1 \sin \m_2)
\\&&
(\sin x_3 \cos \n_2 - \cos x_3 \sin \n_2)
(\cos x_4 \cos \n_3 + \sin x_4 \sin \n_3)
\bigg)^2
\notag\\+&&
\bigg(
(\cos x_2 \cos \n_1 \cos \m_1 + \sin x_2 \cos \n_1 \sin \m_1)
\notag\\&&
(\sin x_3 \cos \n_2 - \cos x_3 \sin \n_2)\sin \n_3
\bigg)^2
\notag\\+&&
\bigg(
(\cos x_2 \sin \n_1 \cos \m_2 + \sin x_2 \sin \n_1 \sin \m_2)
\notag\\&&
(\sin x_3 \cos \n_2 - \cos x_3 \sin \n_2)
(\sin x_4 \cos \n_3 - \cos x_4 \sin \n_3)
\bigg)^2
\notag\\+&&
\bigg(
(\cos x_1 \sin x_2 \cos \n_1 \cos \m_1 - \cos x_1 \cos x_2 \cos \n_1 \sin \m_1 +
\notag\\&&
\sin x_1 \sin x_2 \sin \n_1 \cos \m_2 - \sin x_1 \cos x_2 \sin \n_1 \sin \m_2)
(\sin x_4 \cos \n_3 - \cos x_4 \sin \n_3)\sin \n_2
\bigg)^2
\notag\\+&&
\bigg(
(\sin x_1 \cos \n_1 \sin \m_1 - \cos x_1 \sin \n_1 \sin \m_2)
\notag\\+&&
(\cos x_3 \cos \n_2 + \sin x_3 \sin \n_2)
(\cos y_4 \cos \n_3 + \sin y_4 \sin \n_3)
\bigg)^2
\notag\\+&&
\bigg(
(\cos x_1 \cos \n_1 \sin \m_1 + \sin x_1 \sin \n_1 \sin \m_2)
\notag\\+&&
(\cos x_3 \cos \n_2 + \sin x_3 \sin \n_2)
(\cos x_4 \cos \n_3 + \sin x_4 \sin \n_3)
\bigg)^2
\notag\\+&&
\bigg(
(\sin x_1 \sin x_2 \cos \n_1 \cos \m_1 - \sin x_1 \cos x_2 \cos \n_1 \sin \m_1 -
\notag\\&&
\cos x_1 \sin x_2 \sin \n_1 \cos \m_2 + \cos x_1 \cos x_2 \sin \n_1 \sin \m_2)
(\sin y_4 \cos \n_3 - \cos y_4 \sin \n_3)\sin \n_2
\bigg)^2.
\notag
\end{eqnarray}
Due to the difficulty of simplifying the function $f(\n_1,\n_2,\n_3,\m_1,\m_2)$ in Eq. \eqref{eq:d1d2d3}, we aim to find an upper bound of the function, over the parameters $\n_j,\m_j\in[0,2\p)$ and real constant $x_j,y_j$. In particular, one can obtain that
\begin{eqnarray}
&&
f(0,0,\pi/2,0,\m_2)=
(\cos x_2)^2(\sin x_3)^2,
\label{eq:min1}\\&&
f(0,0,0,\pi/2,\m_2)=
(\cos x_1)^2(\cos x_3)^2(\cos x_4)^2+(\cos x_3)^2(\cos y_4)^2(\sin x_1)^2,
\label{eq:min2}\\&&
f(\p/2,0,\p/2,\p/2,0)=
(\cos x_2)^2(\sin x_3)^2(\cos x_4)^2+(\cos x_2)^2(\sin x_3)^2(\sin x_4)^2,
\notag\\
\label{eq:min3}\\&&
f(\n_1,0,\n_3,0,0)=
(\cos \n_1)^2(\cos \n_3)^2+
(\sin \n_1)^2(\cos x_2)^2(\sin x_3)^2+
\notag\\&&
(\cos \n_1)^2
(\cos x_2)^2
(\sin \n_3)^2
(\sin x_3)^2.
\end{eqnarray}
One can show that 
$f(\n_1,0,\n_3,0,0)\ge f(\n_1,0,\p/2,0,0)=(\cos x_2)^2(\sin x_3)^2$. Hence Eq. \eqref{min:a1nPa1n} is upper bounded by the minimum of \eqref{eq:min1}-\eqref{eq:min3}, denoted as $\cM$. As a result, the geometric measure of entanglement of $\r$, namely $G(\r)$ in \eqref{eq:gme} is upper bounded by $-\log_2(1-\cM )$. So we have managed to evaluate the geometric measure of the PPT entangled state $\r$ from a family of $2\times2\times4$ UPBs we have constructed in this paper. The more explicit evaluation requires a better understanding of \eqref{eq:d1d2d3}.

\section{Conclusions}
\label{sec:con}

We have constructed two families of multipartite UPBs. They are respectively a $2\times2\times4$ UPB of size eight by using the four-qubit UPBs of size eight, as well as a $2\times2\times2\times4$ UPB of size eight by using the five-qubit UPBs of size eight. As a byproduct, we have constructed two novel multipartite PPT entangled states and evaluate their entanglement. A question arising from this paper is to construct more multipartite high-dimensional UPBs using the existing multiqubit UPBs, as it is always a technical challenge with extensive applications in various quantum information tasks.

\section*{Acknowledgments}

LC was supported by the NNSF of China (Grant No. 11871089), and the Fundamental Research Funds for the Central Universities (Grant Nos. KG12040501, ZG216S1810 and ZG226S18C1). 

\section*{Conflict of interest statement}

On behalf of all authors, the corresponding author states that there is no conflict of interest. 

\section*{Data availability statement} 

All data, models, and code generated or used during the study appear in the submitted article.

\appendix

%\section{Appendixes}
%\label{appendix}

\bibliographystyle{unsrt}

\bibliography{channelcontrol}

%\bibliographystyle{plain}
%\bibliography{witness}% Produces the bibliography via BibTeX.

\end{document}